\documentclass[runningheads, envcountsame, a4paper]{llncs}
\usepackage{amsfonts,amsmath}
\usepackage{verbatim,amssymb,amscd}
\usepackage[np]{numprint}

\usepackage{graphicx}
\usepackage[format=plain,labelfont=bf,up,textfont=it,up]{caption}
\usepackage{subfig}

\usepackage{array,multirow,colortbl}%
\usepackage{rotating}%%%
\usepackage{enumerate,textcomp}
\usepackage{xspace}

\usepackage{ifthen}
\usepackage{pgf,tikz}
\usetikzlibrary{arrows,automata,shapes}
\tikzstyle{every state}=[minimum size=12pt,inner sep=0pt]
\usetikzlibrary{snakes}
\tikzstyle{randomPath}=[decorate,decoration={amplitude=1pt,segment length=2pt,random steps},very thick]
\tikzstyle{randomPath2}=[decorate,decoration={amplitude=2pt,segment length=6pt,random steps},very thick]
\usetikzlibrary{calc}

\usepackage{xcolor}
\newcommand{\lacroix}{\tikz[baseline=-.5ex]{\draw[->,>=latex] (0,0) -- (4ex,0); \draw[->,>=latex] (1.8ex,2ex) -- (1.8ex,-2ex);}}
\newcommand{\card}[1]{|#1|}	%{\# #1}
\newcommand{\N}{{\mathbb N}}
\newcommand{\Z}{{\mathbb Z}}

\newcommand{\dz}{\mathfrak d}
\newcommand{\aut}[1]{{\mathcal #1}}
\newcommand{\jjj}{{\aut{J}}}
\newcommand{\dual}[1]{{\mathfrak d}({#1})}

\newcommand{\mot}[1]{{\mathbf {#1}}}

\newcommand{\cc}[1]{\aut{#1}}

\def\eref#1{(\ref{#1})}
\newcommand{\pres}[1]{\langle{#1}\rangle}
\newcommand{\presm}[1]{\pres{{#1}}_{+}}

\def\resp{\hbox{\textit{resp.}} }

\def\Z{\mathbb{Z}}
\def\N{\mathbb{N}}

\newcounter{foo}

\def\dz{{\mathfrak d}}%dualization
\newcommand{\Ac}{\mathcal{A}}
\newcommand{\Bc}{\mathcal{B}}
\newcommand{\resid}[2]{{#2}_{#1}}
\newcommand{\otree}[1][]{\mathfrak{t}{\ifthenelse{\equal{#1}{}}{}{(\aut{#1})}}}

\toctitle{On Torsion-Free Semigroups Generated by~Invertible Reversible Mealy Automata}
\tocauthor{Thibault~Godin, Ines~Klimann, and Matthieu~Picantin}

\begin{document}

\mainmatter 

\title{On Torsion-Free Semigroups Generated by~Invertible Reversible Mealy Automata}
\titlerunning{On Torsion-Free Semigroups Generated by Inv. Rev. Mealy Automata}

\author{Thibault Godin \and Ines Klimann
\and Matthieu Picantin\thanks{The authors are partially supported by the French
\emph{Agence Nationale pour la~Recherche},
through the Project~$\mathbf{MealyM}$ ANR-JCJC-12-JS02-012-01.}}
\authorrunning{Th. Godin, I. Klimann, and M. Picantin}
\institute{Univ Paris Diderot, Sorbonne Paris Cit\'e, LIAFA,
    UMR 7089 CNRS,
    Paris France\\
\email{$\{$godin,klimann,picantin$\}$@liafa.univ-paris-diderot.fr}
}

\maketitle
\setcounter{footnote}{0}

\begin{abstract} This paper addresses the torsion problem for a class of automaton semigroups,
defined as semigroups of transformations induced by Mealy automata,
aka letter-by-letter transducers with the same input and output alphabet.
The torsion problem is undecidable for automaton semigroups in general,
but is known to be solvable within the well-studied class of (semi)groups
generated by invertible bounded Mealy automata.
We focus on the somehow antipodal class of invertible reversible Mealy automata
and prove that for a wide subclass
the generated semigroup is torsion-free.

\keywords{automaton semigroup,
reversible Mealy automaton,
labeled orbit tree, 
torsion-free semigroup
}
\end{abstract}

\section{Introduction}

In this paper we address the torsion problem for a class of automaton semigroups.

\medbreak
In a (semi)group, a \emph{torsion}---or~\emph{periodic}---element is an element of finite order,
that is an element generating a finite monogenic sub(semi)group.
In particular, a (semi)group is \emph{torsion-free} (\resp \emph{torsion})
if its only torsion element is its possible identity element
(\resp if all its elements are torsion elements).
Like most of the major group or semigroup theoretical decision problems,
the word, torsion and finiteness problems are undecidable in general \cite{BBN59}.

\medbreak

Automaton (semi)groups, that is (semi)groups generated by Mealy automata,
were formally introduced a half century ago
(for details, see~\cite{clas32} and references therein). 
Two decades later, important results started revealing their full potential.
In particular, contributing to the so-called Burnside problem, the articles~\cite{aleshin,grigorchuk1}
construct particularly simple Mealy automata generating infinite finitely generated torsion groups,
and, answering the so-called Milnor problem, the articles~\cite{brs,grigorchukMilnor}
describe Mealy automata generating the first examples of (semi)groups with intermediate growth.
Since these pioneering works, a substantial theory continues to develop using various methods,
ranging from finite automata theory to geometric group theory
and never ceases to show that automaton (semi)groups possess multiple interesting
and sometimes unusual features.

\medbreak

For automaton (semi)groups, the word problem is solvable using standard minimization techniques~\cite{bs,eil,KMP12}.
The torsion problem and the finiteness problem for automaton semigroups
have been proven to be undecidable~\cite{gillibert:finiteness14}
but remain open for automaton groups.
However there exist various criteria for recognizing
whether such a (semi)group or one of its element has finite order,
see for instance~\cite{AKLMP12,anto,antoberk,clas32,sidkiconjugacy,cain,Kli13,KPS14:3-state,mal,min,russ,sidki,sst}.
In particular, there are many partial methods to find elements of infinite order in such (semi)groups.
Their efficiency may vary significantly.
By contrast, the class of so-called invertible bounded Mealy automata, which has received considerable attention,
admits an effective solution to both conjugacy and order problems~\cite{BKNamenability,sidkiconjugacy,sidki}.
This class happens to correspond to some tight restriction on the underlying automata:
the non-trivial cycles are disjoint and none can be reached from another.
%none can reach another

\medbreak

Here we tackle the torsion problem, focusing on a very different class of Mealy automata,
namely reversible Mealy automata, in which each connected component turns out to be strongly connected. 
This class was known as the class for which most of the existing partial methods do not work or perform poorly. 
We prove that for a wide subclass of invertible reversible Mealy automata---roughly the non-bireversible
ones---the generated semigroup is torsion-free.
It is worth mentioning that the class of bounded 
Mealy automata and the class of reversible Mealy automata are somehow at the opposite ends of the spectrum.

The proof of torsion-freeness relies on deep structural properties of the so-called \emph{labeled orbit tree}
which happens to capture the behavior of the (strongly) connected components during the exponentiation
of a reversible Mealy automaton, and it gives hopefully
a new insight even in the still mysterious subclass of bireversible Mealy automata
(see~\cite{bs,KPS14:3-state,nek} and the references therein).

\medbreak

The paper is organized as follows. In Section~\ref{sec-mealy}, we set up notation,
provide well-known definitions and facts concerning Mealy automata and automaton semigroups.
Some results concerning connected components of reversible Mealy automata are given in Section~\ref{sec-cc}.
In Section~\ref{sec-otree} we introduce a crucial construction, namely the labeled orbit tree of a Mealy automaton,
and define the notion of a self-liftable path, especially relevant for investigating torsion-freeness.
Finally, Section~\ref{sec-main} contains the proof of our main result.

%-------------------------------------------------------------------------------------
%-------------------------------------------------------------------------------------
\section{Mealy Automata}\label{sec-mealy}
We first recall the formal definition of an automaton.
A {\em (finite, deterministic, and complete) automaton} is a
triple \(\bigl( Q,\Sigma,\delta = (\delta_i\colon Q\rightarrow Q )_{i\in \Sigma} \bigr)\),
where the \emph{stateset}~$Q$
and the \emph{alphabet}~$\Sigma$ are non-empty finite sets, and
where the~\(\delta_i\)
are functions.

\smallskip

A \emph{Mealy automaton} is a quadruple
\(\bigl( Q, \Sigma, \delta = (\delta_i\colon Q\rightarrow Q )_{i\in \Sigma},
\rho = (\rho_x\colon \Sigma\rightarrow \Sigma  )_{x\in Q} \bigr)\),
such that both~\((Q,\Sigma,\delta)\) and~\((\Sigma,Q,\rho)\) are
automata.
In other terms, a Mealy automaton is a complete, deterministic,
letter-to-letter transducer with the same input and output alphabet.

The graphical representation of a Mealy automaton is
standard, see Figures~\ref{fig-lamplighter} and~\ref{fig-jir36}.

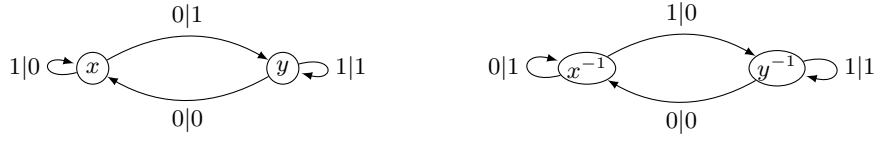
\begin{figure}[ht]%
\centering
\begin{tikzpicture}[->,>=latex,node distance=25mm]
	\node[state] (x) {$x$};
	\node[state] (y) [right of=x] {$y$};
	\path
		(x)	edge[loop left]		node[left]{\(1|0\)}	(x)
		(x)	edge[bend left]		node[above]{\(0|1\)}			(y)
		(y)	edge[loop right] 	node[right]{\(1|1\)}	(y)
		(y)	edge[bend left]		node[below]{\(0|0\)}	(x);
	\begin{scope}[xshift=65mm]
	\node[state,ellipse] (xx) {$x^{-1}$};
	\node[state,ellipse] (yy) [right of=xx] {$y^{-1}$};
	\path
		(xx)	edge[loop left]		node[left]{\(0|1\)}	(xx)
		(xx)	edge[bend left]		node[above]{\(1|0\)}			(yy)
		(yy)	edge[loop right] 	node[right]{\(1|1\)}	(yy)
		(yy)	edge[bend left]		node[below]{\(0|0\)}	(xx);
	\end{scope}
\end{tikzpicture}
\caption{An invertible reversible non-bireversible Mealy automaton~$\aut{L}$ (left) and its inverse~$\aut{L}^{-1}$ (right),
both generating the lamplighter group~$\Z_2\wr\Z$ (see~\cite{gns}).}%
\label{fig-lamplighter}
\end{figure}

\begin{figure}[ht]%
\centering
\begin{tikzpicture}[->,>=latex,node distance=25mm]
		\node[state] (b) {};%$b$};
		\node[state] (a) [right of=b] {};%$a$};
		\node[state] (f) [right of=a] {};%$f$};
		\node[state,node distance=20mm] (c) [below of=b] {};%$c$};
		\node[state] (d) [right of=c] {};%$d$};
		\node[state] (e) [right of=d] {};%$e$};
	\path
		(a)	edge	node[below]{\(1|3\)}	(b)
		(a)	edge[bend right,out=340,in=200]	node[left]{\(2|2\)}	(d)
		(a)	edge	node[below]{\(3|1\)}	(f)
		(b)	edge	node[below, very  near start]{\(3|1\)}	(d)
		(b)	edge[bend left,out=50,in=130]	node{\(\begin{array}{c}1|3\\ 2|2\end{array}\)}	(f)
		(c)	edge	node[above, very near start]{\(1|3\)}	(a)
		(c)	edge	node[left]{\(3|1\)}	(b)
		(c)	edge[bend right, out=335,in=205]	node{\(\begin{array}{c}2|2\\~\end{array}\)}	(e)
		(d)	edge[bend right, out=340,in=200]	node[right]{\(2|2\)}	(a)
		(d)	edge	node[above]{\(3|1\)}	(c)
		(d)	edge	node[above ]{\(1|3\)}	(e)
		(e)	edge	node[above, very near start]{\(3|1\)}	(a)
		(e)	edge[bend left,out=50,in=130]	node{\(\begin{array}{c}1|3\\ 2|2\end{array}\)}	(c)
		(f)	edge[bend right,out=335,in=205]	node{\(\begin{array}{c}\\2|1\end{array}\)}	(b)
		(f)	edge	node[below, very  near start]{\(1|3\)}	(d)
		(f)	edge	node[right]{\(3|2\)}	(e);
\end{tikzpicture}
\caption{A 3-letter 6-state inv. reversible non-bireversible Mealy automaton~$\jjj$. } % 
%m:=MealyMachine([[2,4,6],[6,6,4],[1,5,2],[5,1,3],[3,3,1],[4,2,5]],[[3,2,1],[3,2,1],[3,2,1],[3,2,1],[3,2,1],[3,1,2]]);
\label{fig-jir36}
\end{figure}
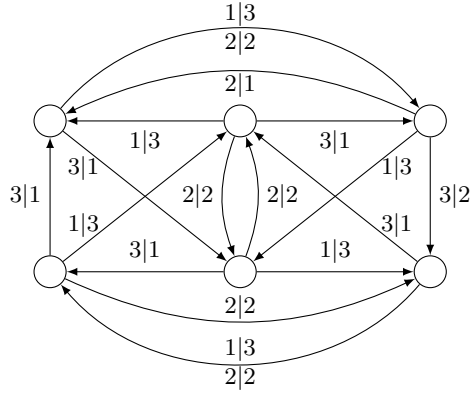

In a Mealy automaton $\aut{A}=(Q,\Sigma, \delta, \rho)$, the sets $Q$ and
$\Sigma$ play dual roles. So we may consider the \emph{dual (Mealy)
  automaton} defined by
\(\dz(\aut{A}) = (\Sigma,Q, \rho, \delta)\).
Alternatively, we can define the dual Mealy automaton via the set of
its transitions:
\begin{equation*}
x \xrightarrow{i\mid j} y \ \in \aut{A} \quad \iff \quad i \xrightarrow{x\mid y} j \ \in
\dz(\aut{A}) \:. 
\end{equation*}

\begin{definition}
A Mealy automaton~$(Q,\Sigma,\delta,\rho)$ is said to be \emph{invertible}
if the functions~$(\rho_x)_{x\in Q}$ are permutations of~$\Sigma$
and~\emph{reversible} if the functions~$(\delta_i)_{i\in\Sigma}$ are permutations of~$Q$.
\end{definition}

Consider a Mealy automaton $\aut{A}=(Q,\Sigma,\delta,\rho)$.
Let $Q^{-1}=\{x^{-1}, x \in Q\}$ be a disjoint copy
of $Q$. The \emph{inverse} \(\aut{A}^{-1}\) of $\aut{A}$ is 
defined by the set of its transitions:
\begin{equation*}
x \xrightarrow{i\mid j} y \ \in \aut{A} \quad \iff \quad x^{-1}
\xrightarrow{j\mid i} y^{-1} \ \in \aut{A}^{-1} \:. 
\end{equation*}

If \(\aut{A}\) is invertible, then its inverse~\(\aut{A}^{-1}\)
 is a Mealy automaton, see for instance Figure~\ref{fig-lamplighter}.

\begin{definition}
A Mealy automaton is \emph{bireversible} if it is invertible, reversible and its inverse is reversible.
\end{definition}

The terms "invertible", "reversible", and "bireversible" are standard since~\cite{MNS}.

Figure~\ref{fig-forbid} gives characterizations of invertibility and reversibility
in terms of forbidden configurations in a Mealy automaton.

\medbreak

Here we define a new class:

\begin{definition}
A Mealy automaton is \emph{coreversible} whenever Configuration~{\bf\small(c)} in Figure~\ref{fig-forbid} does not occur.
This means that each output letter induces a permutation on the stateset.
\end{definition}

The bireversible Mealy automata are those which are simultaneously invertible, reversible, and coreversible.
We emphasize that an invertible reversible Mealy automaton is bireversible if
and only if it is coreversible.

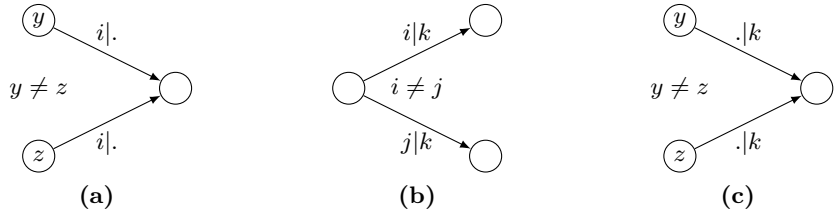
\begin{figure}[ht]%
\centering
\subfloat[]{\label{Config-Rev}
	\begin{tikzpicture}[->,>=latex,node distance=9mm]
	\node[state] (1) {};%{\(x\)};
	\node (0) [left of=1] {};
	\node[state] (2) [left of=0, above of=0] {\(y\)};
	\node[state] (3) [left of=0, below of=0] {\(z\)};
	\node (4) [left of=0] {\(y\neq z\)};
	\path 
      (2) edge node[above]{\(i|.\)} (1)
      (3) edge node[below]{\(i|.\)} (1);
	\end{tikzpicture}
}%
\quad\qquad\qquad
\subfloat[]{\label{Config-Inv}
	\begin{tikzpicture}[->,>=latex,node distance=9mm]
	\node[state] (1) {};%{\(x\)};
	\node (0) [right of=1] {};
	\node[state] (2) [right of=0, above of=0] {};%{\(y\)};
	\node[state] (3) [right of=0, below of=0] {};%{\(z\)};
	\node (4) [right of=1] {\(i\neq j\)};
	\path 
      (1) edge node[above]{\(i|k\)} (2)
      (1) edge node[below]{\(j|k\)} (3);
	\end{tikzpicture}
}
\quad\qquad\qquad
\subfloat[]{\label{Config-Cor}
	\begin{tikzpicture}[->,>=latex,node distance=9mm]
	\node[state] (1) {};%{\(x\)};
	\node (0) [left of=1] {};
	\node[state] (2) [left of=0, above of=0] {\(y\)};
	\node[state] (3) [left of=0, below of=0] {\(z\)};
	\node (4) [left of=0] {\(y\neq z\)};
	\path 
      (2) edge node[above]{\(.|k\)} (1)
      (3) edge node[below]{\(.|k\)} (1);
	\end{tikzpicture}
}
\caption{Configuration~{\bf\small(a)} is forbidden for reversible automata,
Configuration~{\bf\small(b)} for invertible ones,
and Configuration~{\bf\small(c)} for coreversible ones.}
\label{fig-forbid}%
\end{figure}

\medbreak

We view~\(\aut{A}=(Q,\Sigma,\delta,\rho)\) as an automaton with an input and an output tape, thus
defining mappings from input words over~$\Sigma$ to output words
over~$\Sigma$.
Formally, for~\(x\in Q\), the map~$\rho_x\colon\Sigma^* \rightarrow \Sigma^*$,
extending~$\rho_x\colon\Sigma \rightarrow \Sigma$, is defined recursively by:
\begin{equation}\label{eq-rec-def}
\forall i \in \Sigma, \ \forall \mot{s} \in \Sigma^*, \qquad
\rho_x(i\mot{s}) = \rho_x(i)\rho_{\delta_i(x)}(\mot{s}) \:.
\end{equation}

Equation~\eref{eq-rec-def} can be easier to understood if depicted by a \emph{cross-diagram} (see~\cite{AKLMP12}):
\[\begin{array}{ccccc}
		& i             	&            		& \mot{s}       			& \\
x		& \lacroix    	& \delta_i(x) 	& \lacroix				& \delta_{\mot{s}}(\delta_i(x)) \\
		& \rho_x(i)	 	&            		& \rho_{\delta_i(x)}(\mot{s})
\end{array}\]

By convention, the image of the empty word is itself. 
The mapping~\(\rho_x\) for each~$x\in Q$ is length-preserving and prefix-preserving.
We say that~\(\rho_x\) is the \emph{production
function\/} associated with~\((\aut{A},x)\).
For~$\mot{x}=x_1\cdots x_n \in Q^n$ with~$n>0$, set
\(\rho_\mot{x}\colon\Sigma^* \rightarrow \Sigma^*, \rho_\mot{x} = \rho_{x_n}
\circ \cdots \circ \rho_{x_1} \:\).
Denote dually by~\(\delta_i\colon Q^*\rightarrow Q^*,
i\in \Sigma\), the production functions associated with
the dual automaton
$\dz(\aut{A})$. For~$\mot{s}=s_1\cdots s_n
\in \Sigma^n$ with~$n>0$, set~\(\delta_\mot{s}\colon Q^* \rightarrow Q^*,
\ \delta_\mot{s} = \delta_{s_n}\circ \cdots \circ \delta_{s_1}\).

The semigroup of mappings from~$\Sigma^*$ to~$\Sigma^*$ generated by
$\{\rho_x, x\in Q\}$ is called the \emph{semigroup generated
  by~$\aut{A}$} and is denoted by~$\presm{\aut{A}}$.
When~\(\aut{A}\) is invertible,
its production functions are
permutations on words of the same length and thus we may consider
the group of mappings from~$\Sigma^*$ to~$\Sigma^*$ generated by
$\{\rho_x, x\in Q\}$. This group is called the \emph{group generated
by~$\aut{A}$} and is denoted by~$\pres{\aut{A}}$.\medskip

It is know from~\cite{AKLMP12} that the possible behaviors
of invertible reversible non-bireversible Mealy automata provide less variety
than those of bireversible automata whenever finiteness is concerned:

\begin{proposition} \emph{\textbf{(\cite[Corollary 22]{AKLMP12})}}
Any invertible reversible non-bireversible Mealy automaton
generates an infinite group.
\end{proposition}

Note that the ratio of these invertible reversible non-bireversible Mealy automata
tends to supersede the bireversible one, when the size of alphabet and/or stateset increases.

%-------------------------------------------------------------------------------------
%-------------------------------------------------------------------------------------
\section{On the Behavior of Connected Components}\label{sec-cc}

In this section, we gather some properties satisfied by the connected components
of the underlying graph of a reversible Mealy automaton
and we focus on those properties preserved when making products.
We use the following crucial property: any connected component of a reversible Mealy automaton
is strongly connected.
Our main tool, described in the next section, captures the behavior
of the connected components of the successive powers
of a given reversible Mealy automaton, allowing a much finer analysis.

\begin{definition}\label{def-product}
Let $\Ac = (Q,\Sigma, \delta, \rho)$
and $\Bc = (Q',\Sigma, \delta', \rho')$ be two Mealy automata acting on the same alphabet.
Their \emph{product} is the Mealy automaton $\Ac \times \Bc  = (Q \times Q', \Sigma, \gamma, \pi)$ with transition
\[xy\xrightarrow{i| \rho'_y{(\rho_x (i))}} \delta_i(x)\delta'_{\rho_x(i)}(y)\:,\]
which can be seen in terms of cross-diagram as:
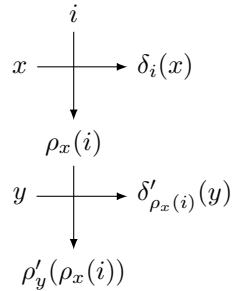
\begin{center}
\begin{tikzpicture}[->,>=latex,node distance=1cm,scale=1]
	\node (ii) [] {$i$};
  	\node (xx) [below left of=ii] {$x$};
	\node (jj) [below right of=xx,below] {$\rho_x(i)$};
	\node (yy) [below right of=ii,right] {$\delta_i(x)$};
	\node (s) [below left of=jj] {$y$};
	\node (z) [below right of=s,below] {$\rho'_y(\rho_x(i))$};
	\node (t) [below right of=jj,right] {$ \delta'_{\rho_x(i)}(y)$};
  	\path 	(ii) edge (jj)
			(xx) edge (yy)
			(s) edge (t)
			(jj) edge (z);
\end{tikzpicture}
\end{center}
\end{definition}

Note that the product of two reversible (\resp invertible) Mealy automata
is still a reversible (\resp invertible) Mealy automaton.
Let us consider the coreversibility property.

\begin{lemma}\label{lem-product}
Let $\Ac$ and $\Bc$ be Mealy automata on the same alphabet with $\Ac$ connected and reversible.
Then, for any connected component~$\cc{C}$ of~$\Ac \times \Bc$,
every state of~$\aut{A}$ occurs as a prefix of some state of~$\cc{C}$.
\end{lemma}

\begin{proof}
Let $\Ac = (Q,\Sigma, \delta, \rho)$
and let~\(\cc{C}\) be a connected component of~\(\Ac\times \Bc\).
Let~$xx'\in\cc{C}$ and~$y\in Q$. Since \(\aut{A}\) is connected and reversible,
there exists~$\mot{s}\in~\Sigma^*$ satisfying~$y=\delta_{\mot{s}}(x)$,
hence $y$ is a prefix of the state~$\delta_{\mot{s}}(xx')$ in~$\cc{C}$.
\qed\end{proof}

\begin{proposition}\label{prop-product}
Let $\Ac$ and $\Bc$ be reversible Mealy automata on the same alphabet.
If $\Ac$ is connected and non-coreversible,
then every connected component of~$\Ac \times \Bc$ is reversible and non-coreversible.
\end{proposition}

\begin{proof}
Let $Q$ be the stateset of~$\Ac$ 
and let~\(\cc{C}\) be a connected component of~\(\Ac\times \Bc\).
As \(\aut{A}\) and \(\aut{B}\) are reversible, so is \(\cc{C}\).

Since $\Ac$ is not coreversible,
there exist two states~$x\neq y \in Q$ leading to the same state~$z$,
when producing the same letter~$j$:

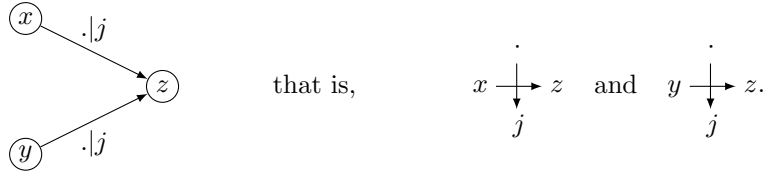
\begin{center}
  \begin{tikzpicture}[->,>=latex,node distance=9mm]
	\node[state] (1) {$z$};%{\(x\)};
	\node (0) [left of=1] {};
	\node[state] (2) [left of=0, above of=0] {\(x\)};
	\node[state] (3) [left of=0, below of=0] {\(y\)};
	\path 
      (2) edge node[above]{\(.|j\)} (1)
      (3) edge node[below]{\(.|j\)} (1);
	\node[node distance=20mm] (4) [right of=1] {that is,};
	\node[node distance=40mm] (5) [right of=4] {\(\begin{array}{ccc}
	&\:\cdot~\\
x	&\lacroix	& z \\
	&j 
\end{array}\)\quad and\quad \(\begin{array}{ccc}
	&\:\cdot~\\
y	&\lacroix	& z. \\
	&j 
\end{array}\)};
\end{tikzpicture}
\end{center}

By Lemma~\ref{lem-product}, \(\cc{C}\) admits a state prefixed with~$x$, say~$xx'$.
Let
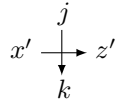
\[\begin{array}{ccc}
	& j  \\
x' 	& \lacroix	& z' \\
	& k     
\end{array}\]

be a transition of \(\aut{B}\), then the following configuration occurs in~\(\cc{C}\):
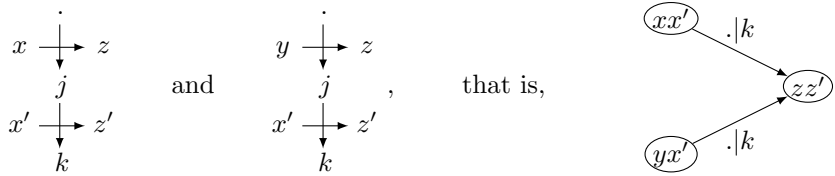
\begin{center}
  \begin{tikzpicture}[->,>=latex,node distance=9mm]
	\node[state,ellipse] (1) {\(zz'\)};%{\(x\)};
	\node (0) [left of=1] {};
	\node[state,ellipse] (2) [left of=0, above of=0] {\(xx'\)};
	\node[state,ellipse] (3) [left of=0, below of=0] {\(yx'\)};
	\path 
      (2) edge node[above]{\(.|k\)} (1)
      (3) edge node[below]{\(.|k\)} (1);
	\node[node distance=40mm] (4) [left of=1] {that is,};
	\node[node distance=40mm] (5) [left of=4] {\(\begin{array}{ccc}
	&\:.~ \\
x	&\lacroix	& z \\
	&j \\
x' 	& \lacroix	& z' \\
	& k     
\end{array}\qquad \text{and}\qquad\begin{array}{ccc}
	&\:.~ \\
y	&\lacroix	& z \\
	&j \\
x' 	& \lacroix	& z' \\
	& k     
\end{array}\:,\)};
\end{tikzpicture}
\end{center}

which means that \(\cc{C}\) cannot be coreversible.
\qed\end{proof}

A convenient and natural operation is to raise a Mealy automaton to some power.
The \emph{\(n\)-th power} of~$\aut{A} = (Q,\Sigma, \delta, \rho)$ is recursively defined.
By convention, $\aut{A}^0$ is the trivial Mealy automaton with only one state, which acts like identity on~$\Sigma$. For~$n>0$,
$\aut{A}^n$ is the Mealy automaton
\begin{equation*}
\aut{A}^n = \bigl( \ Q^n,\Sigma, (\delta_i : Q^n \rightarrow Q^n)_{i\in \Sigma}, (\rho_{\mot{u}} : \Sigma \rightarrow \Sigma)_{\mot{u}\in Q^n} \ \bigr)\ .
\end{equation*}

\begin{corollary}\label{cor-product}
If a Mealy automaton is (invertible) reversible without co\-revers\-ible connected component,
then every connected component of any of its powers is (invertible) reversible and non-coreversible.
\end{corollary}

\begin{definition}\label{def-sph-tr}\cite{gns,nek}
The action of a Mealy automaton~$\aut{A}$ is said to be \emph{spherically transitive}
or~\emph{level-transitive} whenever all the powers of~$\dual{\aut{A}}$ are connected.
\end{definition}

%-------------------------------------------------------------------------------------
%-------------------------------------------------------------------------------------
\section{The Labeled Orbit Tree}\label{sec-otree}

There exist strong links between the successive sizes of the connected components
of the powers of a reversible Mealy automaton
and some finiteness properties of the generated semigroup,
as emphasized by the two following results.
Such links can be captured by a suitable tree, playing a fundamental role in the sequel.

\begin{proposition}\label{prop-bounded-cc}
A reversible Mealy automaton generates a finite semigroup
if and only if the sizes of the connected components of its powers are bounded.
\end{proposition}

The latter is proven in~\cite{KPS14:3-state}
within the framework of invertible reversible Mealy automata,
but the invertibility is not invoked in the proof.
%The following extends another result from~\cite{KPS14:3-state}.
We need the following result, also from~\cite{KPS14:3-state}.

\begin{proposition}\label{prop-order}
Let $\Ac=(Q,\Sigma,\delta,\rho)$ be an invertible reversible Mealy automaton. For any~$\mot{u} \in Q^+$, the following are equivalent:
\begin{enumerate}[(i)]
\item the action~$\rho_{\mot{u}}$ induced by~$\mot{u}$ has finite order;				\label{i1}
\item the sizes of the connected components of~$(\mot{u}^n)_{n \in \N}$ are bounded.	\label{i2}
\end{enumerate}
\end{proposition}

%\begin{proof}
%\eqref{i1}\(\Rightarrow\)\eqref{i2}  is a direct consequence of
%Proposition~\ref{prop-bounded-cc}: let \(k\) be the order of
%the action~$\rho_{\mot{u}}$ induced by~\({\mot{u}}\); it means that~\(\mot{u}^k\) acts as the identity,
%and so do all the states of its connected component.
%By Proposition~\ref{prop-bounded-cc}, the connected components of the
%\((\mot{u}^{kn})_n\) have bounded size, which leads to~\eqref{i2}.
%
%\eqref{i2}\(\Rightarrow\)\eqref{i1}: for each \(n\), denote by
%\(\aut{C}_n\) the connected component of~\(\mot{u}^n\). As the sizes
%of these components are bounded, the sequence \((\aut{C}_n)_n\) admits a
%subsequence whose all elements are the same, up to state numbering.
%Within this subsequence,
%there are two elements such that two different words in~\(\mot{u}^*\)
%name the same state, say~\(\mot{u}^p\) and~\(\mot{u}^q\) with~$p\neq q$, which means
%that~\(\rho_{\mot{u}^p}=\rho_{\mot{u}^q}\) holds, and \(\rho_{\mot{u}}\)
%has finite order.
%\qed\end{proof}

A direct consequence of Proposition~\ref{prop-bounded-cc} %prop-order}
provides a simple yet interesting result concerning torsion-freeness.

\begin{corollary}\label{cor-sph-tr}
Let $\Ac$ be a reversible Mealy automaton. Whenever the action of~$\dual{\aut{A}}$ is spherically transitive,
the semigroup~$\presm{\aut{A}}$ is torsion-free.
\end{corollary}

\begin{proof}
Let~\(\aut{A}\) be a Mealy automaton with stateset~$Q$ such that all its powers are connected.
By Proposition~\ref{prop-bounded-cc}, \(\aut{A}\) generates an infinite semigroup.

Assume that there exists~$\mot{u} \in Q^+$ whose action has finite order, say~\(\rho_{\mot{u}^p}=\rho_{\mot{u}^q}\) with~\(p<q\). 
%Let~$\cc{C} $ be its connected component. 
%All the powers of~$\cc{C} $ are connected. Hence by Proposition~\ref{prop-order},
%$\mot{u}$ induces an action of  infinite order.
By reversibility of~\(\aut{A}\), every state of~$\aut{A}^q$ is equivalent to some state of~$\aut{A}^p$,
hence~\(\aut{A}\) generates a finite semigroup, which is a contradiction.
\qed\end{proof}
 
Corollary~\ref{cor-sph-tr} applies for instance to the Mealy automaton~$\aut{L}$ on Figure~\ref{fig-lamplighter}(left):
the subsemigroup of the lamplighter group generated by~$x$ and~$y$ is torsion-free.

\bigbreak
We are now ready to introduce our main tool.
\begin{definition}
Let~$\Ac$ be a reversible Mealy automaton with stateset~$Q$. 
Rooted in~$\Ac^0$, the \emph{labeled orbit tree}~$\otree[\Ac]$ is constructed
as the graph of the (strongly) connected components of the powers of~$\Ac$,
with an edge between two nodes~$\cc{C}, \cc{D} $
whenever there is $\mot{u} \in \cc{C}$ with~$\mot{u}x \in \cc{D}$ and~$x\in Q$,
such an edge being labeled by the (integer) ratio~$\card{\cc{D}} / \card{\cc{C}}$.
\end{definition}

Such a tree~$\otree[\Ac]$ is more precisely named the labeled orbit tree of the dual~$\dz(\Ac)$
since it can be seen as the tree of the orbits of~$Q^*$ under the action of the group~$\pres{\dz(\Ac)}$
(see~\cite{gawron_ns:conjugation,KPS14:3-state}).

\medbreak
Figure~\ref{fig-otree} displays the labeled orbit tree~$\otree[\jjj]$,
where~$\jjj$ is the Mealy automaton defined in Figure~\ref{fig-jir36}.

%\input{otJIR36bw.tex}
%m:=MealyMachine([[2,4,6],[6,6,4],[1,5,2],[5,1,3],[3,3,1],[4,2,5]],[[3,2,1],[3,2,1],[3,2,1],[3,2,1],[3,2,1],[3,1,2]]);

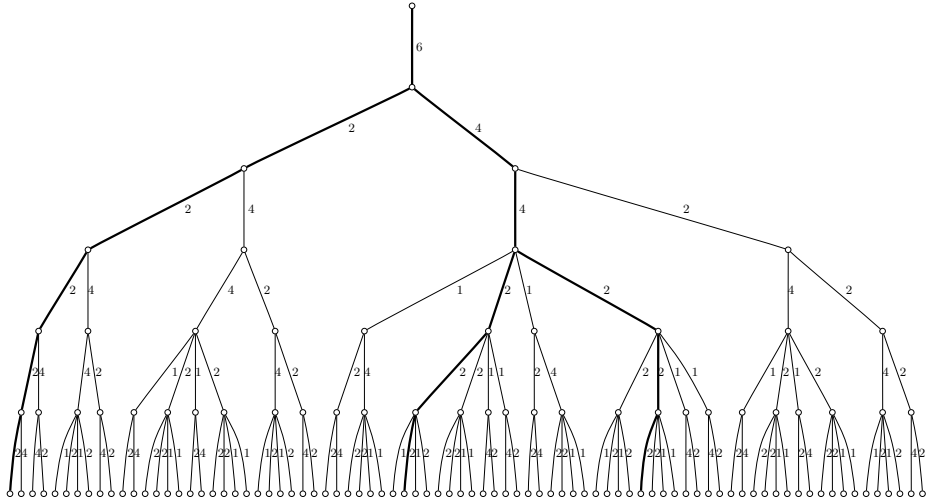
\begin{figure}[ht]
\begin{center}
%\hspace*{-25pt}
\scalebox{.527}{
\begin{tikzpicture}[>=latex',line join=bevel,xscale=.613,inner sep=0pt, minimum size=4pt]
  \node (a6_23) at (289.0bp,4.0bp) [draw,circle] {}; %,ellipse] {};
  \node (a6_22) at (276.0bp,4.0bp) [draw,circle] {}; %,ellipse] {};
  \node (a6_21) at (263.0bp,4.0bp) [draw,circle] {}; %,ellipse] {};
  \node (a6_20) at (250.0bp,4.0bp) [draw,circle] {}; %,ellipse] {};
  \node (a6_27) at (341.0bp,4.0bp) [draw,circle] {}; %,ellipse] {};
  \node (a6_26) at (328.0bp,4.0bp) [draw,circle] {}; %,ellipse] {};
  \node (a6_25) at (315.0bp,4.0bp) [draw,circle] {}; %,ellipse] {};
  \node (a6_24) at (302.0bp,4.0bp) [draw,circle] {}; %,ellipse] {};
  \node (a6_29) at (367.0bp,4.0bp) [draw,circle] {}; %,ellipse] {};
  \node (a6_28) at (354.0bp,4.0bp) [draw,circle] {}; %,ellipse] {};
  \node (a6_38) at (484.0bp,4.0bp) [draw,circle] {}; %,ellipse] {};
  \node (a6_39) at (497.0bp,4.0bp) [draw,circle] {}; %,ellipse] {};
  \node (a6_34) at (432.0bp,4.0bp) [draw,circle] {}; %,ellipse] {};
  \node (a6_35) at (445.0bp,4.0bp) [draw,circle] {}; %,ellipse] {};
  \node (a6_36) at (458.0bp,4.0bp) [draw,circle] {}; %,ellipse] {};
  \node (a6_37) at (471.0bp,4.0bp) [draw,circle] {}; %,ellipse] {};
  \node (a6_30) at (380.0bp,4.0bp) [draw,circle] {}; %,ellipse] {};
  \node (a6_31) at (393.0bp,4.0bp) [draw,circle] {}; %,ellipse] {};
  \node (a6_32) at (406.0bp,4.0bp) [draw,circle] {}; %,ellipse] {};
  \node (a6_33) at (419.0bp,4.0bp) [draw,circle] {}; %,ellipse] {};
  \node (a6_49) at (627.0bp,4.0bp) [draw,circle] {}; %,ellipse] {};
  \node (a6_48) at (614.0bp,4.0bp) [draw,circle] {}; %,ellipse] {};
  \node (a0_1) at (467.0bp,352.0bp) [draw,circle] {}; %,ellipse] {};
  \node (a6_41) at (523.0bp,4.0bp) [draw,circle] {}; %,ellipse] {};
  \node (a6_40) at (510.0bp,4.0bp) [draw,circle] {}; %,ellipse] {};
  \node (a6_43) at (549.0bp,4.0bp) [draw,circle] {}; %,ellipse] {};
  \node (a6_42) at (536.0bp,4.0bp) [draw,circle] {}; %,ellipse] {};
  \node (a6_45) at (575.0bp,4.0bp) [draw,circle] {}; %,ellipse] {};
  \node (a6_44) at (562.0bp,4.0bp) [draw,circle] {}; %,ellipse] {};
  \node (a6_47) at (601.0bp,4.0bp) [draw,circle] {}; %,ellipse] {};
  \node (a6_46) at (588.0bp,4.0bp) [draw,circle] {}; %,ellipse] {};
  \node (a6_4) at (42.0bp,4.0bp) [draw,circle] {}; %,ellipse] {};
  \node (a6_5) at (55.0bp,4.0bp) [draw,circle] {}; %,ellipse] {};
  \node (a6_6) at (68.0bp,4.0bp) [draw,circle] {}; %,ellipse] {};
  \node (a6_7) at (81.0bp,4.0bp) [draw,circle] {}; %,ellipse] {};
  \node (a6_1) at (3.0bp,4.0bp) [draw,circle] {}; %,ellipse] {};
  \node (a6_2) at (16.0bp,4.0bp) [draw,circle] {}; %,ellipse] {};
  \node (a6_3) at (29.0bp,4.0bp) [draw,circle] {}; %,ellipse] {};
  \node (a6_8) at (94.0bp,4.0bp) [draw,circle] {}; %,ellipse] {};
  \node (a6_9) at (107.0bp,4.0bp) [draw,circle] {}; %,ellipse] {};
  \node (a6_58) at (744.0bp,4.0bp) [draw,circle] {}; %,ellipse] {};
  \node (a6_59) at (757.0bp,4.0bp) [draw,circle] {}; %,ellipse] {};
  \node (a6_52) at (666.0bp,4.0bp) [draw,circle] {}; %,ellipse] {};
  \node (a6_53) at (679.0bp,4.0bp) [draw,circle] {}; %,ellipse] {};
  \node (a6_50) at (640.0bp,4.0bp) [draw,circle] {}; %,ellipse] {};
  \node (a6_51) at (653.0bp,4.0bp) [draw,circle] {}; %,ellipse] {};
  \node (a6_56) at (718.0bp,4.0bp) [draw,circle] {}; %,ellipse] {};
  \node (a6_57) at (731.0bp,4.0bp) [draw,circle] {}; %,ellipse] {};
  \node (a6_54) at (692.0bp,4.0bp) [draw,circle] {}; %,ellipse] {};
  \node (a6_55) at (705.0bp,4.0bp) [draw,circle] {}; %,ellipse] {};
  \node (a2_1) at (273.0bp,236.0bp) [draw,circle] {}; %,ellipse] {};
  \node (a2_2) at (586.0bp,236.0bp) [draw,circle] {}; %,ellipse] {};
  \node (a1_1) at (467.0bp,294.0bp) [draw,circle] {}; %,ellipse] {};
  \node (a5_28) at (1043.0bp,62.0bp) [draw,circle] {}; %,ellipse] {};
  \node (a5_22) at (809.0bp,62.0bp) [draw,circle] {}; %,ellipse] {};
  \node (a5_23) at (848.0bp,62.0bp) [draw,circle] {}; %,ellipse] {};
  \node (a5_20) at (751.0bp,62.0bp) [draw,circle] {}; %,ellipse] {};
  \node (a5_21) at (783.0bp,62.0bp) [draw,circle] {}; %,ellipse] {};
  \node (a5_26) at (952.0bp,62.0bp) [draw,circle] {}; %,ellipse] {};
  \node (a5_27) at (1010.0bp,62.0bp) [draw,circle] {}; %,ellipse] {};
  \node (a5_24) at (887.0bp,62.0bp) [draw,circle] {}; %,ellipse] {};
  \node (a5_25) at (913.0bp,62.0bp) [draw,circle] {}; %,ellipse] {};
  \node (a5_3) at (81.0bp,62.0bp) [draw,circle] {}; %,ellipse] {};
  \node (a5_2) at (36.0bp,62.0bp) [draw,circle] {}; %,ellipse] {};
  \node (a5_1) at (16.0bp,62.0bp) [draw,circle] {}; %,ellipse] {};
  \node (a5_7) at (217.0bp,62.0bp) [draw,circle] {}; %,ellipse] {};
  \node (a5_6) at (185.0bp,62.0bp) [draw,circle] {}; %,ellipse] {};
  \node (a5_5) at (146.0bp,62.0bp) [draw,circle] {}; %,ellipse] {};
  \node (a5_4) at (107.0bp,62.0bp) [draw,circle] {}; %,ellipse] {};
  \node (a5_9) at (309.0bp,62.0bp) [draw,circle] {}; %,ellipse] {};
  \node (a5_8) at (250.0bp,62.0bp) [draw,circle] {}; %,ellipse] {};
  \node (a6_67) at (861.0bp,4.0bp) [draw,circle] {}; %,ellipse] {};
  \node (a6_66) at (848.0bp,4.0bp) [draw,circle] {}; %,ellipse] {};
  \node (a6_65) at (835.0bp,4.0bp) [draw,circle] {}; %,ellipse] {};
  \node (a6_64) at (822.0bp,4.0bp) [draw,circle] {}; %,ellipse] {};
  \node (a6_63) at (809.0bp,4.0bp) [draw,circle] {}; %,ellipse] {};
  \node (a6_62) at (796.0bp,4.0bp) [draw,circle] {}; %,ellipse] {};
  \node (a6_61) at (783.0bp,4.0bp) [draw,circle] {}; %,ellipse] {};
  \node (a6_60) at (770.0bp,4.0bp) [draw,circle] {}; %,ellipse] {};
  \node (a3_1) at (93.0bp,178.0bp) [draw,circle] {}; %,ellipse] {};
  \node (a3_3) at (586.0bp,178.0bp) [draw,circle] {}; %,ellipse] {};
  \node (a3_2) at (273.0bp,178.0bp) [draw,circle] {}; %,ellipse] {};
  \node (a3_4) at (901.0bp,178.0bp) [draw,circle] {}; %,ellipse] {};
  \node (a6_69) at (887.0bp,4.0bp) [draw,circle] {}; %,ellipse] {};
  \node (a6_68) at (874.0bp,4.0bp) [draw,circle] {}; %,ellipse] {};
  \node (a5_17) at (608.0bp,62.0bp) [draw,circle] {}; %,ellipse] {};
  \node (a5_16) at (575.0bp,62.0bp) [draw,circle] {}; %,ellipse] {};
  \node (a5_15) at (555.0bp,62.0bp) [draw,circle] {}; %,ellipse] {};
  \node (a5_14) at (523.0bp,62.0bp) [draw,circle] {}; %,ellipse] {};
  \node (a5_13) at (471.0bp,62.0bp) [draw,circle] {}; %,ellipse] {};
  \node (a5_12) at (412.0bp,62.0bp) [draw,circle] {}; %,ellipse] {};
  \node (a5_11) at (380.0bp,62.0bp) [draw,circle] {}; %,ellipse] {};
  \node (a5_10) at (341.0bp,62.0bp) [draw,circle] {}; %,ellipse] {};
  \node (a4_10) at (1010.0bp,120.0bp) [draw,circle] {}; %,ellipse] {};
  \node (a5_19) at (705.0bp,62.0bp) [draw,circle] {}; %,ellipse] {};
  \node (a5_18) at (640.0bp,62.0bp) [draw,circle] {}; %,ellipse] {};
  \node (a6_70) at (900.0bp,4.0bp) [draw,circle] {}; %,ellipse] {};
  \node (a6_71) at (913.0bp,4.0bp) [draw,circle] {}; %,ellipse] {};
  \node (a6_72) at (926.0bp,4.0bp) [draw,circle] {}; %,ellipse] {};
  \node (a6_73) at (939.0bp,4.0bp) [draw,circle] {}; %,ellipse] {};
  \node (a6_74) at (952.0bp,4.0bp) [draw,circle] {}; %,ellipse] {};
  \node (a6_75) at (965.0bp,4.0bp) [draw,circle] {}; %,ellipse] {};
  \node (a6_76) at (978.0bp,4.0bp) [draw,circle] {}; %,ellipse] {};
  \node (a6_77) at (991.0bp,4.0bp) [draw,circle] {}; %,ellipse] {};
  \node (a6_78) at (1004.0bp,4.0bp) [draw,circle] {}; %,ellipse] {};
  \node (a6_79) at (1017.0bp,4.0bp) [draw,circle] {}; %,ellipse] {};
  \node (a6_16) at (198.0bp,4.0bp) [draw,circle] {}; %,ellipse] {};
  \node (a6_17) at (211.0bp,4.0bp) [draw,circle] {}; %,ellipse] {};
  \node (a6_14) at (172.0bp,4.0bp) [draw,circle] {}; %,ellipse] {};
  \node (a6_15) at (185.0bp,4.0bp) [draw,circle] {}; %,ellipse] {};
  \node (a6_12) at (146.0bp,4.0bp) [draw,circle] {}; %,ellipse] {};
  \node (a6_13) at (159.0bp,4.0bp) [draw,circle] {}; %,ellipse] {};
  \node (a6_10) at (120.0bp,4.0bp) [draw,circle] {}; %,ellipse] {};
  \node (a6_11) at (133.0bp,4.0bp) [draw,circle] {}; %,ellipse] {};
  \node (a6_18) at (224.0bp,4.0bp) [draw,circle] {}; %,ellipse] {};
  \node (a6_19) at (237.0bp,4.0bp) [draw,circle] {}; %,ellipse] {};
  \node (a6_82) at (1056.0bp,4.0bp) [draw,circle] {}; %,ellipse] {};
  \node (a6_81) at (1043.0bp,4.0bp) [draw,circle] {}; %,ellipse] {};
  \node (a6_80) at (1030.0bp,4.0bp) [draw,circle] {}; %,ellipse] {};
  \node (a4_8) at (751.0bp,120.0bp) [draw,circle] {}; %,ellipse] {};
  \node (a4_9) at (901.0bp,120.0bp) [draw,circle] {}; %,ellipse] {};
  \node (a4_2) at (93.0bp,120.0bp) [draw,circle] {}; %,ellipse] {};
  \node (a4_3) at (217.0bp,120.0bp) [draw,circle] {}; %,ellipse] {};
  \node (a4_1) at (36.0bp,120.0bp) [draw,circle] {}; %,ellipse] {};
  \node (a4_6) at (555.0bp,120.0bp) [draw,circle] {}; %,ellipse] {};
  \node (a4_7) at (608.0bp,120.0bp) [draw,circle] {}; %,ellipse] {};
  \node (a4_4) at (309.0bp,120.0bp) [draw,circle] {}; %,ellipse] {};
  \node (a4_5) at (412.0bp,120.0bp) [draw,circle] {}; %,ellipse] {};
  \draw [] (a3_4) ..controls (918.77bp,167.87bp) and (992.57bp,129.96bp)  .. (a4_10);
  \draw (971.5bp,149.0bp) node {2};
  \draw [] (a5_20) ..controls (752.31bp,48.773bp) and (755.73bp,16.821bp)  .. (a6_59);
  \draw (757.5bp,33.0bp) node {1};
  \draw [] (a5_18) ..controls (640.0bp,48.773bp) and (640.0bp,16.821bp)  .. (a6_50);
  \draw (643.5bp,33.0bp) node {2};
  \draw [] (a5_20) ..controls (749.52bp,54.312bp) and (747.91bp,46.618bp)  .. (747.0bp,40.0bp) .. controls (745.34bp,27.917bp) and (744.44bp,13.351bp)  .. (a6_58);
  \draw (750.5bp,33.0bp) node {2};
  \draw [] (a5_8) ..controls (256.04bp,55.186bp) and (263.12bp,47.71bp)  .. (267.0bp,40.0bp) .. controls (272.64bp,28.79bp) and (274.96bp,13.734bp)  .. (a6_22);
  \draw (276.5bp,33.0bp) node {1};
  \draw [] (a5_19) ..controls (705.0bp,48.773bp) and (705.0bp,16.821bp)  .. (a6_55);
  \draw (708.5bp,33.0bp) node {1};
  \draw [] (a5_27) ..controls (1008.7bp,48.773bp) and (1005.3bp,16.821bp)  .. (a6_78);
  \draw (1010.5bp,33.0bp) node {2};
  \draw [] (a5_8) ..controls (247.58bp,54.707bp) and (244.72bp,46.859bp)  .. (243.0bp,40.0bp) .. controls (240.0bp,28.04bp) and (238.03bp,13.405bp)  .. (a6_19);
  \draw (246.5bp,33.0bp) node {2};
  \draw [] (a5_6) ..controls (182.58bp,54.707bp) and (179.72bp,46.859bp)  .. (178.0bp,40.0bp) .. controls (175.0bp,28.04bp) and (173.03bp,13.405bp)  .. (a6_14);
  \draw (181.5bp,33.0bp) node {2};
  \draw [] (a5_13) ..controls (471.0bp,48.773bp) and (471.0bp,16.821bp)  .. (a6_37);
  \draw (474.5bp,33.0bp) node {1};
  \draw [ultra thick] (a3_3) ..controls (608.89bp,169.23bp) and (727.77bp,128.88bp)  .. (a4_8);
  \draw (691.5bp,149.0bp) node {2};
  \draw [] (a4_4) ..controls (315.64bp,107.38bp) and (334.47bp,74.423bp)  .. (a5_10);
  \draw (332.5bp,91.0bp) node {2};
  \draw [] (a4_6) ..controls (548.36bp,107.38bp) and (529.53bp,74.423bp)  .. (a5_14);
  \draw (545.5bp,91.0bp) node {2};
  \draw [] (a5_12) ..controls (410.69bp,48.773bp) and (407.27bp,16.821bp)  .. (a6_32);
  \draw (412.5bp,33.0bp) node {2};
  \draw [] (a5_15) ..controls (556.53bp,48.773bp) and (560.52bp,16.821bp)  .. (a6_44);
  \draw (562.5bp,33.0bp) node {2};
  \draw [] (a5_28) ..controls (1043.0bp,48.773bp) and (1043.0bp,16.821bp)  .. (a6_81);
  \draw (1046.5bp,33.0bp) node {4};
  \draw [] (a4_7) ..controls (614.64bp,107.38bp) and (633.47bp,74.423bp)  .. (a5_18);
  \draw (630.5bp,91.0bp) node {4};
  \draw [] (a5_21) ..controls (783.0bp,48.773bp) and (783.0bp,16.821bp)  .. (a6_61);
  \draw (786.5bp,33.0bp) node {4};
  \draw [] (a5_17) ..controls (609.31bp,48.773bp) and (612.73bp,16.821bp)  .. (a6_48);
  \draw (614.5bp,33.0bp) node {4};
  \draw [] (a5_6) ..controls (179.39bp,55.089bp) and (172.78bp,47.539bp)  .. (169.0bp,40.0bp) .. controls (163.34bp,28.701bp) and (160.42bp,13.694bp)  .. (a6_13);
  \draw (172.5bp,33.0bp) node {2};
  \draw [] (a5_3) ..controls (81.0bp,48.773bp) and (81.0bp,16.821bp)  .. (a6_7);
  \draw (84.5bp,33.0bp) node {1};
  \draw [ultra thick] (a5_20) ..controls (746.58bp,55.309bp) and (740.96bp,47.499bp)  .. (738.0bp,40.0bp) .. controls (733.45bp,28.472bp) and (731.74bp,13.594bp)  .. (a6_57);
  \draw (741.5bp,33.0bp) node {2};
  \draw [] (a5_14) ..controls (525.42bp,54.707bp) and (528.28bp,46.859bp)  .. (530.0bp,40.0bp) .. controls (533.0bp,28.04bp) and (534.97bp,13.405bp)  .. (a6_42);
  \draw (535.5bp,33.0bp) node {1};
  \draw [] (a5_5) ..controls (146.0bp,48.773bp) and (146.0bp,16.821bp)  .. (a6_12);
  \draw (149.5bp,33.0bp) node {4};
  \draw [] (a4_9) ..controls (911.66bp,107.3bp) and (942.14bp,73.821bp)  .. (a5_26);
  \draw (935.5bp,91.0bp) node {2};
  \draw [] (a4_3) ..controls (217.0bp,106.77bp) and (217.0bp,74.821bp)  .. (a5_7);
  \draw (220.5bp,91.0bp) node {1};
  \draw [] (a5_21) ..controls (785.42bp,54.707bp) and (788.28bp,46.859bp)  .. (790.0bp,40.0bp) .. controls (793.0bp,28.04bp) and (794.97bp,13.405bp)  .. (a6_62);
  \draw (795.5bp,33.0bp) node {2};
  \draw [] (a3_3) ..controls (590.53bp,165.47bp) and (603.28bp,133.02bp)  .. (a4_7);
  \draw (602.5bp,149.0bp) node {1};
  \draw [] (a3_2) ..controls (261.3bp,165.3bp) and (227.82bp,131.82bp)  .. (a4_3);
  \draw (255.5bp,149.0bp) node {~4};
  \draw [] (a5_7) ..controls (218.53bp,48.773bp) and (222.52bp,16.821bp)  .. (a6_18);
  \draw (225.5bp,33.0bp) node {4};
  \draw [] (a5_26) ..controls (952.0bp,48.773bp) and (952.0bp,16.821bp)  .. (a6_74);
  \draw (955.5bp,33.0bp) node {2};
  \draw [] (a5_24) ..controls (884.58bp,54.707bp) and (881.72bp,46.859bp)  .. (880.0bp,40.0bp) .. controls (877.0bp,28.04bp) and (875.03bp,13.405bp)  .. (a6_68);
  \draw (883.5bp,33.0bp) node {2};
  \draw [] (a5_24) ..controls (889.42bp,54.707bp) and (892.28bp,46.859bp)  .. (894.0bp,40.0bp) .. controls (897.0bp,28.04bp) and (898.97bp,13.405bp)  .. (a6_70);
  \draw (899.5bp,33.0bp) node {1};
  \draw [] (a5_4) ..controls (107.0bp,48.773bp) and (107.0bp,16.821bp)  .. (a6_9);
  \draw (110.5bp,33.0bp) node {4};
  \draw [] (a5_9) ..controls (307.52bp,54.312bp) and (305.91bp,46.618bp)  .. (305.0bp,40.0bp) .. controls (303.34bp,27.917bp) and (302.44bp,13.351bp)  .. (a6_24);
  \draw (308.5bp,33.0bp) node {2};
  \draw [] (a5_17) ..controls (606.52bp,54.312bp) and (604.91bp,46.618bp)  .. (604.0bp,40.0bp) .. controls (602.34bp,27.917bp) and (601.44bp,13.351bp)  .. (a6_47);
  \draw (607.5bp,33.0bp) node {2};
  \draw [] (a5_2) ..controls (37.31bp,48.773bp) and (40.733bp,16.821bp)  .. (a6_4);
  \draw (43.5bp,33.0bp) node {2};
  \draw [] (a5_9) ..controls (312.84bp,54.845bp) and (317.38bp,47.104bp)  .. (320.0bp,40.0bp) .. controls (324.32bp,28.305bp) and (326.77bp,13.521bp)  .. (a6_26);
  \draw (327.5bp,33.0bp) node {2};
  \draw [] (a5_14) ..controls (516.96bp,55.186bp) and (509.88bp,47.71bp)  .. (506.0bp,40.0bp) .. controls (500.36bp,28.79bp) and (498.04bp,13.734bp)  .. (a6_39);
  \draw (509.5bp,33.0bp) node {2};
  \draw [] (a5_22) ..controls (811.42bp,54.707bp) and (814.28bp,46.859bp)  .. (816.0bp,40.0bp) .. controls (819.0bp,28.04bp) and (820.97bp,13.405bp)  .. (a6_64);
  \draw (821.5bp,33.0bp) node {2};
  \draw [] (a5_11) ..controls (377.58bp,54.707bp) and (374.72bp,46.859bp)  .. (373.0bp,40.0bp) .. controls (370.0bp,28.04bp) and (368.03bp,13.405bp)  .. (a6_29);
  \draw (376.5bp,33.0bp) node {2};
  \draw [] (a5_18) ..controls (637.58bp,54.707bp) and (634.72bp,46.859bp)  .. (633.0bp,40.0bp) .. controls (630.0bp,28.04bp) and (628.03bp,13.405bp)  .. (a6_49);
  \draw (636.5bp,33.0bp) node {2};
  \draw [ultra thick] (a4_8) ..controls (751.0bp,106.77bp) and (751.0bp,74.821bp)  .. (a5_20);
  \draw (754.5bp,91.0bp) node {2};
  \draw [] (a5_19) ..controls (707.42bp,54.707bp) and (710.28bp,46.859bp)  .. (712.0bp,40.0bp) .. controls (715.0bp,28.04bp) and (716.97bp,13.405bp)  .. (a6_56);
  \draw (717.5bp,33.0bp) node {2};
  \draw [] (a2_2) ..controls (619.79bp,228.99bp) and (867.07bp,185.03bp)  .. (a3_4);
  \draw (783.5bp,207.0bp) node {2};
  \draw [] (a5_27) ..controls (1006.2bp,54.845bp) and (1001.6bp,47.104bp)  .. (999.0bp,40.0bp) .. controls (994.68bp,28.305bp) and (992.23bp,13.521bp)  .. (a6_77);
  \draw (1002.5bp,33.0bp) node {1};
  \draw [] (a4_10) ..controls (1016.8bp,107.38bp) and (1036.3bp,74.423bp)  .. (a5_28);
  \draw (1033.5bp,91.0bp) node {2};
  \draw [] (a3_2) ..controls (280.47bp,165.38bp) and (301.66bp,132.42bp)  .. (a4_4);
  \draw (299.5bp,149.0bp) node {2};
  \draw [] (a5_26) ..controls (949.58bp,54.707bp) and (946.72bp,46.859bp)  .. (945.0bp,40.0bp) .. controls (942.0bp,28.04bp) and (940.03bp,13.405bp)  .. (a6_73);
  \draw (948.5bp,33.0bp) node {2};
  \draw [] (a5_12) ..controls (408.16bp,54.845bp) and (403.62bp,47.104bp)  .. (401.0bp,40.0bp) .. controls (396.68bp,28.305bp) and (394.23bp,13.521bp)  .. (a6_31);
  \draw (404.5bp,33.0bp) node {2};
  \draw [] (a5_15) ..controls (553.69bp,48.773bp) and (550.27bp,16.821bp)  .. (a6_43);
  \draw (555.5bp,33.0bp) node {4};
  \draw [] (a5_18) ..controls (645.61bp,55.089bp) and (652.22bp,47.539bp)  .. (656.0bp,40.0bp) .. controls (661.66bp,28.701bp) and (664.58bp,13.694bp)  .. (a6_52);
  \draw (664.5bp,33.0bp) node {1};
  \draw [] (a5_19) ..controls (698.96bp,55.186bp) and (691.88bp,47.71bp)  .. (688.0bp,40.0bp) .. controls (682.36bp,28.79bp) and (680.04bp,13.734bp)  .. (a6_53);
  \draw (691.5bp,33.0bp) node {1};
  \draw [] (a4_9) ..controls (889.93bp,107.3bp) and (858.24bp,73.821bp)  .. (a5_23);
  \draw (883.5bp,91.0bp) node {1};
  \draw [] (a5_13) ..controls (464.96bp,55.186bp) and (457.88bp,47.71bp)  .. (454.0bp,40.0bp) .. controls (448.36bp,28.79bp) and (446.04bp,13.734bp)  .. (a6_35);
  \draw (457.5bp,33.0bp) node {1};
  \draw [ultra thick] (a3_1) ..controls (82.449bp,166.63bp) and (46.751bp,131.56bp)  .. (a4_1);
  \draw (75.5bp,149.0bp) node {2};
  \draw [] (a5_13) ..controls (473.42bp,54.707bp) and (476.28bp,46.859bp)  .. (478.0bp,40.0bp) .. controls (481.0bp,28.04bp) and (482.97bp,13.405bp)  .. (a6_38);
  \draw (483.5bp,33.0bp) node {2};
  \draw [ultra thick] (a2_1) ..controls (249.15bp,227.58bp) and (117.43bp,186.6bp)  .. (a3_1);
  \draw (208.5bp,207.0bp) node {2};
  \draw [] (a4_6) ..controls (559.12bp,107.47bp) and (570.71bp,75.02bp)  .. (a5_16);
  \draw (570.5bp,91.0bp) node {1};
  \draw [] (a4_9) ..controls (903.62bp,106.77bp) and (910.47bp,74.821bp)  .. (a5_25);
  \draw (911.5bp,91.0bp) node {1};
  \draw [] (a4_3) ..controls (223.85bp,107.38bp) and (243.27bp,74.423bp)  .. (a5_8);
  \draw (241.5bp,91.0bp) node {2};
  \draw [] (a5_3) ..controls (75.392bp,55.089bp) and (68.779bp,47.539bp)  .. (65.0bp,40.0bp) .. controls (59.336bp,28.701bp) and (56.418bp,13.694bp)  .. (a6_5);
  \draw (68.5bp,33.0bp) node {1};
  \draw [] (a3_1) ..controls (93.0bp,164.77bp) and (93.0bp,132.82bp)  .. (a4_2);
  \draw (96.5bp,149.0bp) node {4};
  \draw [] (a5_11) ..controls (380.0bp,48.773bp) and (380.0bp,16.821bp)  .. (a6_30);
  \draw (383.5bp,33.0bp) node {4};
  \draw [] (a2_1) ..controls (273.0bp,222.77bp) and (273.0bp,190.82bp)  .. (a3_2);
  \draw (276.5bp,207.0bp) node {~~4};
  \draw [] (a5_2) ..controls (34.519bp,54.312bp) and (32.911bp,46.618bp)  .. (32.0bp,40.0bp) .. controls (30.336bp,27.917bp) and (29.441bp,13.351bp)  .. (a6_3);
  \draw (35.5bp,33.0bp) node {4};
  \draw [] (a5_23) ..controls (845.58bp,54.707bp) and (842.72bp,46.859bp)  .. (841.0bp,40.0bp) .. controls (838.0bp,28.04bp) and (836.03bp,13.405bp)  .. (a6_65);
  \draw (844.5bp,33.0bp) node {2};
  \draw [] (a4_3) ..controls (203.86bp,108.63bp) and (159.39bp,73.563bp)  .. (a5_5);
  \draw (193.5bp,91.0bp) node {1};
  \draw [] (a3_3) ..controls (561.77bp,169.2bp) and (435.42bp,128.54bp)  .. (a4_5);
  \draw (522.5bp,149.0bp) node {1};
  \draw [] (a5_22) ..controls (809.0bp,48.773bp) and (809.0bp,16.821bp)  .. (a6_63);
  \draw (812.5bp,33.0bp) node {4};
  \draw [] (a3_4) ..controls (901.0bp,164.77bp) and (901.0bp,132.82bp)  .. (a4_9);
  \draw (904.5bp,149.0bp) node {4};
  \draw [] (a5_4) ..controls (109.42bp,54.707bp) and (112.28bp,46.859bp)  .. (114.0bp,40.0bp) .. controls (117.0bp,28.04bp) and (118.97bp,13.405bp)  .. (a6_10);
  \draw (120.5bp,33.0bp) node {2};
  \draw [] (a4_8) ..controls (741.39bp,107.3bp) and (713.89bp,73.821bp)  .. (a5_19);
  \draw (736.5bp,91.0bp) node {2};
  \draw [] (a4_8) ..controls (757.64bp,107.38bp) and (776.47bp,74.423bp)  .. (a5_21);
  \draw (773.5bp,91.0bp) node {1};
  \draw [] (a5_8) ..controls (250.0bp,48.773bp) and (250.0bp,16.821bp)  .. (a6_20);
  \draw (253.5bp,33.0bp) node {2};
  \draw [ultra thick] (a2_2) ..controls (586.0bp,222.77bp) and (586.0bp,190.82bp)  .. (a3_3);
  \draw (589.5bp,207.0bp) node {~~4};
  \draw [] (a5_14) ..controls (520.58bp,54.707bp) and (517.72bp,46.859bp)  .. (516.0bp,40.0bp) .. controls (513.0bp,28.04bp) and (511.03bp,13.405bp)  .. (a6_40);
  \draw (519.5bp,33.0bp) node {2};
  \draw [] (a5_6) ..controls (187.42bp,54.707bp) and (190.28bp,46.859bp)  .. (192.0bp,40.0bp) .. controls (195.0bp,28.04bp) and (196.97bp,13.405bp)  .. (a6_16);
  \draw (198.5bp,33.0bp) node {1};
  \draw [ultra thick] (a4_1) ..controls (31.881bp,107.47bp) and (20.293bp,75.02bp)  .. (a5_1);
  \draw (32.5bp,91.0bp) node {2};
  \draw [] (a4_2) ..controls (90.38bp,106.77bp) and (83.533bp,74.821bp)  .. (a5_3);
  \draw (92.5bp,91.0bp) node {4};
  \draw [] (a5_7) ..controls (215.69bp,48.773bp) and (212.27bp,16.821bp)  .. (a6_17);
  \draw (218.5bp,33.0bp) node {2};
  \draw [] (a4_5) ..controls (405.36bp,107.38bp) and (386.53bp,74.423bp)  .. (a5_11);
  \draw (403.5bp,91.0bp) node {2};
  \draw [] (a5_16) ..controls (575.0bp,48.773bp) and (575.0bp,16.821bp)  .. (a6_45);
  \draw (578.5bp,33.0bp) node {4};
  \draw [] (a4_7) ..controls (608.0bp,106.77bp) and (608.0bp,74.821bp)  .. (a5_17);
  \draw (611.5bp,91.0bp) node {2};
  \draw [] (a5_10) ..controls (341.0bp,48.773bp) and (341.0bp,16.821bp)  .. (a6_27);
  \draw (344.5bp,33.0bp) node {4};
  \draw [] (a5_18) ..controls (642.42bp,54.707bp) and (645.28bp,46.859bp)  .. (647.0bp,40.0bp) .. controls (650.0bp,28.04bp) and (651.97bp,13.405bp)  .. (a6_51);
  \draw (652.5bp,33.0bp) node {1};
  \draw [] (a4_4) ..controls (309.0bp,106.77bp) and (309.0bp,74.821bp)  .. (a5_9);
  \draw (312.5bp,91.0bp) node {4};
  \draw [] (a4_10) ..controls (1010.0bp,106.77bp) and (1010.0bp,74.821bp)  .. (a5_27);
  \draw (1013.5bp,91.0bp) node {4};
  \draw [] (a5_19) ..controls (702.58bp,54.707bp) and (699.72bp,46.859bp)  .. (698.0bp,40.0bp) .. controls (695.0bp,28.04bp) and (693.03bp,13.405bp)  .. (a6_54);
  \draw (701.5bp,33.0bp) node {2};
  \draw [] (a5_27) ..controls (1011.5bp,48.773bp) and (1015.5bp,16.821bp)  .. (a6_79);
  \draw (1017.5bp,33.0bp) node {1};
  \draw [] (a5_6) ..controls (185.0bp,48.773bp) and (185.0bp,16.821bp)  .. (a6_15);
  \draw (188.5bp,33.0bp) node {1};
  \draw [ultra thick] (a5_13) ..controls (468.58bp,54.707bp) and (465.72bp,46.859bp)  .. (464.0bp,40.0bp) .. controls (461.0bp,28.04bp) and (459.03bp,13.405bp)  .. (a6_36);
  \draw (467.5bp,33.0bp) node {2};
  \draw [] (a4_1) ..controls (36.0bp,106.77bp) and (36.0bp,74.821bp)  .. (a5_2);
  \draw (39.5bp,91.0bp) node {4};
  \draw [] (a4_5) ..controls (412.0bp,106.77bp) and (412.0bp,74.821bp)  .. (a5_12);
  \draw (415.5bp,91.0bp) node {4};
  \draw [] (a5_16) ..controls (577.42bp,54.707bp) and (580.28bp,46.859bp)  .. (582.0bp,40.0bp) .. controls (585.0bp,28.04bp) and (586.97bp,13.405bp)  .. (a6_46);
  \draw (587.5bp,33.0bp) node {2};
  \draw [] (a4_6) ..controls (555.0bp,106.77bp) and (555.0bp,74.821bp)  .. (a5_15);
  \draw (558.5bp,91.0bp) node {1};
  \draw [] (a5_12) ..controls (413.53bp,48.773bp) and (417.52bp,16.821bp)  .. (a6_33);
  \draw (419.5bp,33.0bp) node {1};
  \draw [] (a5_10) ..controls (343.42bp,54.707bp) and (346.28bp,46.859bp)  .. (348.0bp,40.0bp) .. controls (351.0bp,28.04bp) and (352.97bp,13.405bp)  .. (a6_28);
  \draw (354.5bp,33.0bp) node {2};
  \draw [] (a5_28) ..controls (1045.4bp,54.707bp) and (1048.3bp,46.859bp)  .. (1050.0bp,40.0bp) .. controls (1053.0bp,28.04bp) and (1055.0bp,13.405bp)  .. (a6_82);
  \draw (1055.5bp,33.0bp) node {2};
  \draw [] (a5_9) ..controls (304.78bp,54.898bp) and (299.81bp,47.198bp)  .. (297.0bp,40.0bp) .. controls (292.47bp,28.386bp) and (290.14bp,13.557bp)  .. (a6_23);
  \draw (300.5bp,33.0bp) node {1};
  \draw [] (a4_9) ..controls (897.94bp,106.77bp) and (889.96bp,74.821bp)  .. (a5_24);
  \draw (898.5bp,91.0bp) node {2};
  \draw [] (a5_26) ..controls (957.61bp,55.089bp) and (964.22bp,47.539bp)  .. (968.0bp,40.0bp) .. controls (973.66bp,28.701bp) and (976.58bp,13.694bp)  .. (a6_76);
  \draw (976.5bp,33.0bp) node {1};
  \draw [ultra thick] (a1_1) ..controls (441.2bp,285.55bp) and (298.12bp,244.25bp)  .. (a2_1);
  \draw (397.5bp,265.0bp) node {2};
  \draw [] (a4_2) ..controls (96.057bp,106.77bp) and (104.04bp,74.821bp)  .. (a5_4);
  \draw (105.5bp,91.0bp) node {2};
  \draw [] (a5_3) ..controls (78.584bp,54.707bp) and (75.72bp,46.859bp)  .. (74.0bp,40.0bp) .. controls (71.0bp,28.04bp) and (69.025bp,13.405bp)  .. (a6_6);
  \draw (77.5bp,33.0bp) node {2};
  \draw [] (a5_24) ..controls (880.96bp,55.186bp) and (873.88bp,47.71bp)  .. (870.0bp,40.0bp) .. controls (864.36bp,28.79bp) and (862.04bp,13.734bp)  .. (a6_67);
  \draw (873.5bp,33.0bp) node {2};
  \draw [ultra thick] (a0_1) ..controls (467.0bp,338.77bp) and (467.0bp,306.82bp)  .. (a1_1);
  \draw (470.5bp,323.0bp) node {~~6};
  \draw [] (a5_5) ..controls (143.58bp,54.707bp) and (140.72bp,46.859bp)  .. (139.0bp,40.0bp) .. controls (136.0bp,28.04bp) and (134.03bp,13.405bp)  .. (a6_11);
  \draw (142.5bp,33.0bp) node {2};
  \draw [] (a5_23) ..controls (848.0bp,48.773bp) and (848.0bp,16.821bp)  .. (a6_66);
  \draw (851.5bp,33.0bp) node {4};
  \draw [] (a5_25) ..controls (915.42bp,54.707bp) and (918.28bp,46.859bp)  .. (920.0bp,40.0bp) .. controls (923.0bp,28.04bp) and (924.97bp,13.405bp)  .. (a6_72);
  \draw (925.5bp,33.0bp) node {4};
  \draw [] (a4_3) ..controls (210.36bp,107.38bp) and (191.53bp,74.423bp)  .. (a5_6);
  \draw (208.5bp,91.0bp) node {2};
  \draw [ultra thick] (a5_1) ..controls (13.584bp,54.707bp) and (10.72bp,46.859bp)  .. (9.0bp,40.0bp) .. controls (6.0001bp,28.04bp) and (4.0254bp,13.405bp)  .. (a6_1);
  \draw (12.5bp,33.0bp) node {2};
  \draw [ultra thick] (a1_1) ..controls (486.41bp,283.87bp) and (566.97bp,245.96bp)  .. (a2_2);
  \draw (543.5bp,265.0bp) node {4};
  \draw [ultra thick] (a3_3) ..controls (579.57bp,165.38bp) and (561.32bp,132.42bp)  .. (a4_6);
  \draw (577.5bp,149.0bp) node {2};
  \draw [] (a5_24) ..controls (887.0bp,48.773bp) and (887.0bp,16.821bp)  .. (a6_69);
  \draw (890.5bp,33.0bp) node {1};
  \draw [] (a5_12) ..controls (415.79bp,54.825bp) and (420.28bp,47.068bp)  .. (423.0bp,40.0bp) .. controls (427.5bp,28.288bp) and (430.46bp,13.514bp)  .. (a6_34);
  \draw (430.5bp,33.0bp) node {1};
  \draw [ultra thick] (a4_6) ..controls (540.31bp,109.21bp) and (485.19bp,72.457bp)  .. (a5_13);
  \draw (525.5bp,91.0bp) node {2};
  \draw [] (a5_20) ..controls (754.41bp,54.777bp) and (758.47bp,46.985bp)  .. (761.0bp,40.0bp) .. controls (765.28bp,28.204bp) and (768.37bp,13.477bp)  .. (a6_60);
  \draw (768.5bp,33.0bp) node {1};
  \draw [] (a4_8) ..controls (757.88bp,113.96bp) and (768.72bp,105.79bp)  .. (777.0bp,98.0bp) .. controls (789.23bp,86.493bp) and (802.29bp,71.1bp)  .. (a5_22);
  \draw (793.5bp,91.0bp) node {1};
  \draw [] (a5_8) ..controls (252.42bp,54.707bp) and (255.28bp,46.859bp)  .. (257.0bp,40.0bp) .. controls (260.0bp,28.04bp) and (261.97bp,13.405bp)  .. (a6_21);
  \draw (263.5bp,33.0bp) node {1};
  \draw [] (a5_9) ..controls (310.31bp,48.773bp) and (313.73bp,16.821bp)  .. (a6_25);
  \draw (316.5bp,33.0bp) node {1};
  \draw [] (a5_27) ..controls (1013.8bp,54.825bp) and (1018.3bp,47.068bp)  .. (1021.0bp,40.0bp) .. controls (1025.5bp,28.288bp) and (1028.5bp,13.514bp)  .. (a6_80);
  \draw (1028.5bp,33.0bp) node {2};
  \draw [] (a5_25) ..controls (913.0bp,48.773bp) and (913.0bp,16.821bp)  .. (a6_71);
  \draw (916.5bp,33.0bp) node {2};
  \draw [] (a5_14) ..controls (523.0bp,48.773bp) and (523.0bp,16.821bp)  .. (a6_41);
  \draw (526.5bp,33.0bp) node {1};
  \draw [] (a5_1) ..controls (16.0bp,48.773bp) and (16.0bp,16.821bp)  .. (a6_2);
  \draw (19.5bp,33.0bp) node {4};
  \draw [] (a5_26) ..controls (954.42bp,54.707bp) and (957.28bp,46.859bp)  .. (959.0bp,40.0bp) .. controls (962.0bp,28.04bp) and (963.97bp,13.405bp)  .. (a6_75);
  \draw (964.5bp,33.0bp) node {1};
  \draw [] (a5_3) ..controls (83.416bp,54.707bp) and (86.28bp,46.859bp)  .. (88.0bp,40.0bp) .. controls (91.0bp,28.04bp) and (92.975bp,13.405bp)  .. (a6_8);
  \draw (94.5bp,33.0bp) node {2};
\end{tikzpicture}
}
\end{center}
\caption{The labeled orbit tree~$\otree[\jjj]$ (up to level~6) with~$\jjj$ defined on Figure~\ref{fig-jir36}
(the~thickened edges emphasize the 1-self-liftable paths defined below in Def.~\ref{def-self-liftable}).}
\label{fig-otree}
\end{figure}

\medbreak
Since the orbit trees are rooted,
we choose the classical orientation where the root is the higher vertex and the tree grows from top to bottom.
A path is a (possibly infinite) sequence of adjacent edges without backtracking.
The initial vertex of an edge $e$ is denoted by~$\top (e)$ and its terminal vertex by~$\bot(e)$;
by extension, the initial vertex of a non-empty path $\mot{e}$ is denoted by~$\top(\mot{e})$
and its terminal vertex by~$\bot(\mot{e})$ whenever the path is finite.
The label of a (possibly infinite) path is the ordered sequence of labels of its edges.

\begin{definition}
In a tree, a (possibly infinite) path $\mot{e}$ is said to be \emph{initial} if $\top(\mot{e})$ is the root of the tree.
\end{definition}

\begin{definition}
Let~\(\aut{A}\) be a reversible Mealy automaton and let~\(e,f\) be two edges in the orbit tree~\(\otree[A]\).
We say that~\(e\) \emph{is liftable to}~\(f\) if each word of~\(\bot(e)\)
admits some word of~\(\bot(f)\) as a suffix.
\end{definition}

One can notice that this condition is not as strong as it seems: 

\begin{lemma}\label{lem-lift-one}
Let~\(\aut{A}\) be a reversible Mealy automaton and let~\(e,f\) be two edges in the orbit tree~\(\otree[A]\).
If there exists a word of~$\bot(e)$ which admits a word of~$\bot(f)$ as suffix,
then $e$ is liftable to~$f$.
\end{lemma}

\begin{proof} Assume $\mot{u}\mot{v} \in \bot(e)$ with~$\mot{v} \in \bot(f) $.
By reversibility, for any word~$\mot{w}$ in the connected component~$\bot(e)$,
there exists~$\mot{s} \in \Sigma^*$ satisfying~$\mot{w}= \delta_{\mot{s}}(\mot{u}\mot{v}) $,
which can also be written $\mot{w} = \delta_{\mot{s}}(\mot{u})\delta_{\mot{t}}(\mot{v}) $
with~$\mot{t} = \rho_{\mot{u}}(\mot{s})$.
Hence the suffix~$\delta_{\mot{t}}(\mot{v})$ of~$\mot{w}$
belongs to the connected component~$\bot(f)$ of~$\mot{v}$.
\qed\end{proof}

\begin{definition}\label{def-self-liftable}
Let~\(\aut{A}\) be a reversible Mealy automaton
and let~$\mot{e}$ be a (possibly infinite) initial path in the orbit tree~\(\otree[A]\).
We say that $\mot{e}=e_0e_1\cdots$ is \emph{1-self-liftable}
whenever every edge~$e_{i+1}$ is liftable to its predecessor~$e_i$, for~$i\geq 0$.
\end{definition}

This important notion generalizes that of an~\emph{$e$-liftable path} used in~\cite{KPS14:3-state}
where the liftability is required with respect to a uniquely specified edge~$e$.

Using thickened edges, Figure~\ref{fig-otree} highlights each of the %three
1-self-liftable paths in the orbit tree~$\otree[\jjj]$,
where the Mealy automaton~$\jjj$ is displayed on Figure~\ref{fig-jir36}.

\begin{definition}\label{def-path-of}
Let~\(\aut{A}\) be a reversible Mealy automaton with stateset~$Q$.
The \emph{path of} a word~$\mot{u}\in Q^*\cup Q^\omega$ 
is the unique initial path in~$\otree[A]$ going from the root through the connected
components of the prefixes of~\(\mot{u}\).
\end{definition}

\begin{lemma}\label{lem-omega}
Let~\(\aut{A}\) be a reversible Mealy automaton.
For any state~$x$ of~$\aut{A}$, the path of~$x^{\omega}$ in~$\otree[A]$ is 1-self-liftable.
\end{lemma}

\begin{proof}
By Lemma~\ref{lem-lift-one}, $x^n$ being a suffix of~$x^{n+1}$, such a path is 1-self-liftable.
\qed\end{proof}

Lemma~\ref{lem-omega} guarantees the existence of 1-self-liftable paths in any orbit tree.

%-------------------------------------------------------------------------------------
%-------------------------------------------------------------------------------------
\section{Main Result}\label{sec-main}

Assume that $\Ac$
is an invertible reversible Mealy automaton without bireversible component.
Our aim is to prove that every element of~$\presm{\aut{A}}$ has infinite order. 
We first prove this property for the states of $\aut{A}$, whenever~$\Ac$ is connected,
by looking at some %well-chosen
1-self-liftable paths in~$\otree[A]$ (defined in Section~\ref{sec-otree}).
Then we extend it to arbitrary elements of~$\presm{\aut{A}}$ by using
the properties of products of Mealy automata (established in Section~\ref{sec-cc}).

\begin{proposition}\label{prop-no-one}
Let $\Ac$ be some connected invertible reversible non-bireversible Mealy automaton.
A 1-self-liftable path in~$\otree[A]$ cannot contain an edge labeled by~1.
\end{proposition}

\newcommand{\LLL}{T}%{L}
\newcommand{\LLLp}{L}%{L'}

\begin{proof}
Let~$\Ac= (Q,\Sigma, \delta, \rho)$, $\mot{e}$ be a 1-self-liftable path in~$\otree[A]$ and $e$ be an edge of~$\mot{e}$.
Let $\LLL$ (\emph{resp.} $\LLLp$) denote the set of states of~$\top(e)$ (\emph{resp.} of~$\bot(e)$).
For any word~\(\mot{w}\), \(\resid{\mot{w}}{L}=\{\mot{u}\mid \mot{wu}\in L\}\)
is the \emph{left quotient} of~\(L\) by~\(\mot{w}\) (see for instance~\cite{saka}).

As $\Ac$ is connected and reversible, according to Lemma~\ref{lem-product}, 
for any $x \in Q$,
the left quotient~$\resid{x}{\LLLp}$ is non-empty. 
Hence we have\[\LLLp= \bigsqcup_{x \in Q} x\resid{x}{\LLLp}\quad\hbox{(disjoint union)}.\]
The hypotheses that the path~$\mot{e}$ is 1-self-liftable
and that~$\aut{A}$ is invertible yield~\[\LLL= \bigcup_{x \in Q} \resid{x}{\LLLp}.\]
Indeed, let~$x\mot{u}\in\LLLp$ with~$x\in Q$ 
(and~$\mot{u}\in\LLL$ by 1-self-liftability) and let~$\mot{v}\in\LLL$.
By reversibility, there exists~$\mot{s}\in\Sigma^*$ verifying~$\delta_{\mot{s}}(\mot{u})=\mot{v}$.
Now, by invertibility, there exists~$\mot{t}\in\Sigma^*$ with~$\rho_{x}(\mot{t})=\mot{s}$:
\[\begin{array}{ccl}
		& {\color{gray}\mot{t}}  \\
x		& {\color{gray}\lacroix} 	& {\color{gray}\delta_{\mot{t}}(x)}=x'\\
		& \mot{s}\\
\mot{u}	& \lacroix				& \delta_{\mot{s}}(\mot{u})=\mot{v}\\
		& \rho_{\mot{u}}(\mot{s}).	
\end{array}\]Therefore, $\mot{v}$ is a suffix of~$\delta_{\mot{t}}(x\mot{u})$, hence~$\mot{v}\in L_{x'}$ for~$x'=\delta_{\mot{t}}(x)\in Q$.

Since $\Ac$ is not coreversible, there exist $y \neq y',z \in Q$ and $i,j,k \in \Sigma$ satisfying
\[\begin{array}{cccm{1cm}ccc}
	& i      	&  	&			&   	& k \\
y	& \lacroix 	& z 	& \text{ and }  	& y'	& \lacroix & z\\
	& j 		& 	& 			&	& j
\end{array}.\]
So, in the connected component $\bot(e)$, we have
\[\delta_i(y\resid{y}{\LLLp})=z\resid{z}{\LLLp}\quad\hbox{and}\quad\delta_k(y'\resid{y'}{\LLLp})=z\resid{z}{\LLLp}.\]
From reversibility of~$\Ac$, $\delta_j$ is injective and we deduce~$\resid{y}{\LLLp}=\resid{y'}{\LLLp}$. 
Therefore the union~$\LLL=\bigcup_{x \in Q} \resid{x}{\LLLp}$ is not disjoint
and we find~$\card{\LLL} < \card{\LLLp}$ which implies that the label of~$e$ is greater than~1.
\qed\end{proof}

An easy but interesting first consequence is the following.

\begin{corollary}\label{cor-conn3free}
A connected 3-state invertible reversible non-bireversible Mealy automaton generates a free semigroup.
\end{corollary}

\begin{proof} We deduce from Proposition~\ref{prop-no-one} that any connected 3-state invertible reversible
non-bireversible Mealy automaton sees its dual to be spherically transitive
and the result follows from~\cite[Proposition~14]{Kli13}.
\qed\end{proof}

\medbreak
Let us go back to our main purpose.

\begin{proposition}\label{prop-order-gene}
Let $\Ac$ be some connected invertible reversible non-bireversible Mealy automaton.
Then any state of~$\Ac$ induces an action of infinite order.
\end{proposition} 

\begin{proof}
Let~$x$ be a state of~$\aut{A}$. The path of~$x^\omega$ is 1-self-liftable by Lemma~\ref{lem-omega}.
So by Proposition~\ref{prop-no-one}, this path has no edge labeled with~$1$, which means that
the sizes of the connected components of~$(x^n)_{n \in \N}$ are unbounded.
By Proposition~\ref{prop-order}, the action induced by~$x$ has infinite order.
\qed\end{proof}

We can now state our main result by extending Proposition~\ref{prop-order-gene}:
\begin{theorem}\label{thm-main}
Any invertible reversible Mealy automaton without bireversible component
generates a torsion-free semigroup.
\end{theorem}

\begin{proof}
Let~\(\aut{A}\) be an invertible reversible Mealy automaton without bireversible component with stateset~$Q$.
Let $\mot{u} \in Q^+$ and let~$\cc{C} $ its connected component in~$\otree[\Ac]$.
From Corollary~\ref{cor-product}, $\cc{C} $ is a connected invertible reversible non-bireversible Mealy automaton
with~$\mot{u}$ as a state.
Hence by Proposition~\ref{prop-order-gene}, $\mot{u}$ induces an action of  infinite order.
\qed\end{proof}

Note that Theorem~\ref{thm-main} cannot provide extra information on the torsion-freeness of the generated group.
Take for instance the Mealy automaton~$\aut{L}$ of Figure~\ref{fig-lamplighter} (left): the action induced by~$yx^{-1}$ has order~2.
However, Theorem~\ref{thm-main} ensures that an invertible reversible Mealy automaton without bireversible component
cannot generate an infinite Burnside group (see~\cite{nek} for background on the Burnside problem).

\bigskip
All these results and constructions emphasize the relevance of the reversibility property
and question us further on those (semi)groups structures generated by bireversible automata that,
despite the tightness of the hypothesis on them,
reveal more complex to study.

\section*{Acknowledgments}
The authors thank an anonymous referee, whose relevant comments improved the paper.

\bibliographystyle{splncs03}

\bibliography{inv-rev-nonbirev}
\end{document}